\documentclass[11pt]{article}
\usepackage{amsthm}
\usepackage{amsfonts}
\usepackage{amssymb}
\usepackage{graphicx}
\usepackage{caption}
\usepackage{subcaption}
\usepackage{amsmath}
\usepackage{lscape}

\usepackage{fullpage}
\usepackage{tikz}
\newtheorem{Theorem}{Theorem}
\newtheorem{Definition}{Definition}
\newtheorem{Corollary}{Corollary}
\newtheorem{Lemma}{Lemma}
\newtheorem{Example}{Example}
\newtheorem{Proposition}{Proposition}
\newtheorem{Remark}{Remark}
\begin{document}
\title{\bf A  Shock Model Based  Approach to Network Reliability}
\author{
\textbf{S. Zarezadeh\footnote{Department of Statistics, Shiraz University, Shiraz 71454, Iran  (E-mail: s.zarezadeh@hotmail.com)}, \ \ S. Ashrafi\footnote{Department of Statistics, University of Isfahan, Isfahan 81744, Iran (E-mail: s.ashrafi@sci.ui.ac.ir)}\ \ and\ \ M. Asadi\footnote{{Department of Statistics, University of Isfahan, Isfahan 81744, Iran \& School of  Mathematics, Institute of Research in Fundamental Sciences (IPM), P.O Box 19395-5746, Tehran, Iran} (E-mail: m.asadi@sci.ui.ac.ir)}}\\
}
\date{}

\maketitle
\begin{abstract}
 We consider a network consisting of $n$ components (links or nodes) and assume that the network has two states, {\it up } and {\it down}. We further  suppose  that the network  is subject to shocks that appear according to a counting process and that each shock may lead to the component  failures. Under some assumptions on the shock occurrences,  we present a new variant of the notion of {\it signature} which we call it {\it t-signature}. Then t-signature based mixture representations for the reliability function of the network are obtained.  Several stochastic properties of the network lifetime are  investigated. In particular, under the assumption that the number of failures at each shock follows a binomial distribution and the  process of shocks is non-homogeneous Poisson process, explicit form of the network reliability is derived and  its aging properties are explored. Several examples are also provided.
\end{abstract}

\section{Introduction}
Networks include a wide variety  of real-life systems in  communication, industry, software engineering, etc.
A network is defined to be  a collection of {\it nodes} ({\it vertices}) and {\it links} ({\it edges}) in which some particular nodes are called {\it terminals}. For instance,  nodes can be considered as road intersections, telecommunications switches, servers, and computers; and examples of links  can be telecommunication fiber, railways,  copper cable, wireless channels, etc.

\quad According to the existing literature,  a network  can be modeled  by the triplet ${\bf N}=(V,E,T)$, in which $V$ shows the node set, where we assume  $|V|=m$, $E$ stands for  link set, with $|E|=n$, and $T\subseteq V$ is a set of all terminals.
When all terminals of the network are connected to each other, the network is called
$T-$connected.   We assume that the components (links or nodes) of a network are subject to failure, where the failure of the components   may occur according to a stochastic mechanism. A link failure means that the link is obliterated and a node failure means that all links incident to that node are erased. Assuming that the network  has two states {\it up}, and {\it down}, the failure of the components may result in the change of the  state of the network.

\quad In  reliability engineering literature, several approaches are proposed to assess the reliability of a network.  An approach, to study the reliability of a network with $n$ components,  is based on the assumption that the components of the network have statistically independent and identically distributed (i.i.d.) lifetimes $X_1,X_2,\dots, X_n$, and the network has a lifetime $T$ which is a function of $X_1,\dots,X_n$.   An important concept in this approach  is the notion of {\it signature} that is presented in the following definition; see  \cite{saman} and \cite{Gert2}.
\begin{Definition}
{\rm Assume that $\pi=(e_{i_1},e_{i_2},\ldots,e_{i_n})$ is a permutation of the network components  numbers. Suppose that all components in this permutation are up. We move along the permutation, from left to right, and turn the state of each component from up  to down state. Under the assumption that all permutations are equally likely,  the signature vector of the network is defined as ${\bf s}=(s_1,...,s_n)$ where
 $$s_i={n_i}/{n!},\qquad i=1,\dots, n$$  where $n_i$ is the number of permutations in which the failure of $i$th component cause the state of the network changes to a down state. In other words, $s_i$ is the probability that the lifetime of the network equals the $i$th ordered lifetimes among $X_i$'s, i.e., $s_i=P(T=X_{i:n}),$  where $X_{i:n}$ is the $i$th order statistic among the random variables $X_1,X_2,\dots,X_n$. }
\end{Definition}
 The signature vector depends on both the structure of the network and how to define its states. However, it does not depend on the real random mechanism of the component failures. Under this setting, the reliability of the network lifetime $T$,  at time $t>0$, can be represented as
\begin{eqnarray}
P(T>t)=\sum_{i=1}^{n}s_i P(X_{i:n}>t)\label{sigj},
\end{eqnarray}
see \cite{saman1}.  In recent years, a large number of research works are reported in the literature  investigating different  properties of the reliability function (\ref{sigj}). We refer,  among others, to  \cite{saman1}-\cite{Zhang-Li} and references therein.

Another approach, in assessing the reliability of a network, is recently proposed by Gertsbakh and Shpungin \cite{Gert2}. These authors  consider a network with $n$  components, and  assume that  the component failures appear according to a renewal process $\{N(t),t\geq 0\}$  defined as a sequence of i.i.d. non-negative random variables (r.v.s)  $Y_1,Y_2,...,Y_k,\dots$. The random variable $N(t)$ shows the number of  components that fail in the network on interval $[0,t]$, and the failures in $\{N(t),t\geq 0\}$ appear at the instants $S_k=\sum_{i=1}^k{Y_i}, \ \ k=1,2,\dots$. Under the assumption that all orders of component failures are equally likely, the reliability function of the network lifetime $T$ can be represented as
\begin{align}
P(T> t)=\sum_{i=1}^{n}\bar{S}_iP(N(t)=i),\quad t>0,\label{e11}
\end{align}
where $\bar{S}_{i}=\sum_{k=i+1}^n{s_k}, $. Motivated by this, under the assumption that the failure of the network  components  occur according to a counting process,  Zarezadeh and Asadi \cite{zare-ieee} investigated various properties of the model in (\ref{e11}) based on  different scenarios.   Zarezadeh et al. \cite{zare-ejor} studied stochastic properties of  dynamic reliability of networks under the assumption that the components fail according to a nonhomogeneous Poisson process (NHPP).

\quad The aim of the present study is to give new  models for the reliability of the network under the assumption that the components of the network are subject to shocks. We consider a two-state network and assume that the network is subject to shocks that  appear according to a counting process. We further  assume that each shock may lead to component failure and consequently  the network finally fails  by one of the arriving  shocks. The reset of the paper is organized as follows: In Section 2, we obtain the  mixture representations for the reliability of the network lifetime. For this purpose, a new variant of the notion of signature, call it {\it t-signature},  is introduced which allows us to assume that at same time more that one component failure may occur.  We  then compare the t-signature based reliability of two different networks under various assumptions. In Section 3, we assume that the number of failed components in each shock are conditionally distributed as binomial distribution. Under this condition,  mixture representations for the reliability function of the network are obtained and  stochastic and aging properties of the network lifetime are investigated. In particular,  we show that when the shocks arrive according to a non-homogeneous Poisson process (NHPP) and   the arrival time of the first shock has increasing hazard rate average (IHRA), then the distribution of the network lifetime  is IHRA. Section 4 is devoted  to the reliability of the network under fatal shocks. It is assumed that at time of occurrence of a shock at least one component of the network fails. Under  this assumption a mixture representation for the network reliability is obtained based on a new variant of the notion of signature.

\section{Network reliability  under shock models}\label{2}
In this section, we assume that the network is subject to shocks that appear according to a counting process.
In reality, this may happen as a result of a sequence of heavy road accidents, floods, earthquakes, fires etc. We  explore  the reliability of the network   where  each shock  may lead to the failure of the network components. Before doing so,   we define a variant  of the concept of signature which avoids the restriction of not allowing the ties. To be more precise, let $X_1,..., X_n$ be  i.i.d  random variables representing the component lifetimes of the network. One of the assumptions that is necessary  to define the notion of signature  is that there do not exist  ties between $X_1,\dots, X_n$, i.e. $P(X_i=X_j)=0$ for every $i\neq j$ (see, for example,  \cite{sn}). However, in real life situation, this is possible that more than one component may  fail at each time instant,  i.e.  ties may  exist between $X_1,\dots, X_n$. For example when the network is under shock, each shock may results the failure of more that one component at the same time. Under this assumption, in the sequel, we define a variant of the notion of signature. First let us define the discrete random variable $M$ as the minimum number of  components  that their failures cause    the network failure. Obviously $M$ takes values on $\{1,2,\dots,n\}$.
Suppose further that  $n^*$ is the  number of ways that the components fail in the network and $n_i$ is the number of ways of the order of component failures in  which $M=i$. Assuming that all  the number of ways of the order of component failures are equally likely, we define the {\it "tie signature"} ({\it t-signature})  vector associated to the network as ${\bf s}^{\tau}=(s_1^{\tau},\dots, s_n^{\tau})$  where
\begin{align*}
s^{\tau}_i=\frac{n_i }{n^*}, \qquad i=1,...,n.
\end{align*}
It should be noted  that t-signature, similar to the concept of signature,  depends only on the structure of the network and does not depend on the  random mechanism of the component failures.

In the following example, we compute the t-signature vector for a simple network.

\begin{Example}\label{exa} \em
Consider a network with $3$ links and $3$ nodes depicted in Figure \ref{fig1}.
The links are subjected to failure and nodes $a$ and $c$ are considered as terminals. We  assume that the network is functioning if and only if terminals are connected.
\begin{figure}[h!]
  \centering
  \includegraphics[width=3in,height=3in,keepaspectratio]{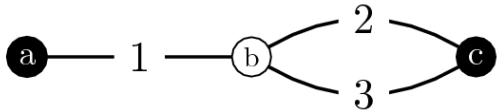}
  \caption{Network with $3$ links and $3$ nodes.}
  \label{fig1}
\end{figure}

Let $\pi$ denote the order of link failures in the network . All possible $\pi$ and the associated  $M$  are presented in Table \ref{tab1}, where the numbers in the braces indicate that the corresponding links failed at the same time. Hence, $n^*=13$ and the elements of the t-signature are calculated as

\begin{table}[]
\begin{center}
\caption{\small All ways of order of links failures }
\label{tab1}
\begin{tabular}{|c|c|c|c|c|c|}
\hline
 $\pi$ & $M$ & $\pi$ & $M$ & $\pi$ & $M$ \\
 \hline
(1,2,3) & 1 &  (\{1,3\},2) & 1 & (\{1,2,3\}) & 1
\\ \hline
(1,3,2) & 1 &  (\{2,3\},1) & 2 &  &
\\ \hline
(2,1,3) & 2 &  (\{1,2\},3) & 1 &  &
\\ \hline
(2,3,1) & 2 &  (3,\{1,2\}) & 2 &  &
\\ \hline
(3,1,2) & 2 &  (2,\{1,3\}) & 2 &  &
\\ \hline
(3,2,1) & 2 &  (1,\{2,3\}) & 1 &  &
\\
\hline
\end{tabular}
\end{center}
\end{table}

\begin{align*}
 s_1^{\tau}=\frac{6}{13}, \qquad s_2^{\tau}=\frac{7}{13},\qquad s_3^{\tau}=0.
\end{align*}
It is interesting to note that the signature vector of this network equals  ${\bf s}=(\frac{1}{3},\frac{2}{3},0).$
\end{Example}

The following lemma gives a formula for computing $n^*$.
\begin{Lemma}\label{n*}
Let a $n$-component network be under shocks. Let $n^*$ be the number of   ways that the components the network fail under the assumption of ties. Then
\begin{align*}
n^*=\sum_{j=1}^{n}\sum_{k=0}^{j}{j \choose k} (-1)^k (j-k)^n.
\end{align*}
\begin{proof}
We use the following combinatorial argument: The number of ways to put $n$ distinct objects into $m$ distinct boxes, $n>m$, such that every box contains at least one object is
$$\sum_{k=0}^{m}{m \choose k} (-1)^k (m-k)^n.$$
Let $J$ be the number of shocks such that in occurrence of each one at least one component fails. It is clear that $J$ takes value on $\{1,...,n\}$. If  $J=j$, is fixed,  the number of ways that  the components numbers $\{1,2,...,n\}$  can be under  $j$ shocks is the same as the number of ways to put $n$ distinct objects into $j$ distinct boxes such that every box contains at least one object. Thus, summing up over $j$, $j=1,\dots,n$, we get
\begin{align*}
n^*=\sum_{j=1}^{n}\sum_{k=0}^{j}{j \choose k} (-1)^k (j-k)^n.
\end{align*}
\end{proof}
\end{Lemma}

\quad Consider a two-state network with lifetime $T$ which is subject to shocks, where shocks  appear according to a counting process, denoted by $\{\xi(t),t>0\}$, at random time instants $\vartheta_1,\vartheta_2,\dots$.  We assume that each shock may lead to component failures and further   assume that the network finally fails  by one of these shocks. Let random variable $W_i$, $i=1,\dots,n$, denote the number of components that fail at the $i$th shock and $W_0\equiv 0$. If $N(t)$ denotes the total number of components that fail up to time $t$, then $N(t)$ takes values on $\{1,2,\dots,n\}$ and
\begin{align*}
  N(t)=\sum_{i=0}^{\xi(t)} W_i.
\end{align*}
Under the assumption that the process of occurrence of the shocks is independent of the number of failed components,  using the law of total probability, the distribution function of  $N(t)$ can be written as
\begin{align}\label{nt}
  P(N(t)\leq x)= &\sum_{k=0}^{\infty} P(N(t) \leq x | \xi(t)=k) P(\xi(t)=k) \nonumber\\
  =&\sum_{k=0}^{\infty} P(\sum_{i=0}^{\xi(t)} W_i \leq x | \xi(t)=k) P(\xi(t)=k) \nonumber\\
    =& \sum_{k=0}^{\infty} H_{k}(x) P(\xi(t)=k),
\end{align}
where $H_k(x)$ denotes  the distribution function of r.v. $\sum_{i=0}^{k} W_i$.
By these assumptions,  the network  fails if $N(t)\geq M$. Hence, the network lifetime can be defined as
\[T\equiv\min_{t>0}\{N(t)\geq M\}\]
and thus,  we have $P(T>t)=P(N(t)<M)$.   Therefore, using the law of total probability and the fact that the total number of components that fail up to time $t$ is independent of the t-signature, we get
\begin{align}\label{M}
P(T>t)=&P(N(t)<M)\nonumber\\
=&\sum_{i=1}^n P(M=i) P(N(t)\leq i-1)\nonumber\\
=&\sum_{i=1}^n s_i^{\tau} P(N(t)\leq i-1).
\end{align}

Let $\bar{S}_j^{\tau}=\sum_{i=j+1}^n s_i^{\tau}$, then using (\ref{nt}) and (\ref{M}), we have
\begin{align}
P(T>t)=&\sum_{i=1}^n s_i^{\tau} \sum_{k=0}^{\infty} H_{k}(i-1) P(\xi(t)=k)\nonumber\\
=&\sum_{k=0}^{\infty} \beta_{k,n} P(\xi(t)=k),\label{eq-org}
\end{align}
where, for $k=0,1,\dots$
\begin{align}\label{b}
\beta_{k,n}=&\sum_{i=1}^n s_i^{\tau} H_{k}(i-1)\nonumber\\
=& \sum_{j=0}^{n-1} \bar{S}_j^{\tau} P(\sum_{i=0}^k W_i=j).
\end{align}

 In the following proposition  some properties of $\beta_{k,n}$ are investigated.
\begin{Proposition}\label{pro}
Let $\vartheta_{1},\vartheta_{2},...$ be the epoch times corresponding to $\{\xi(t),t>0\}$. Then
 $$\beta_{k,n}=P(T>\vartheta_k),$$
 and as a function of $k$, $\beta_{k,n}$  is a survival function  with probability mass function  ${\bf b}_n=(b_{1,n},b_{2,n},...)$, where
 $ b_{k,n}=P(T=\vartheta_k).$
\end{Proposition}
\begin{proof}
We have
\begin{align}\label{eee}
P(T>\vartheta_k)&=\sum_{m=1}^{n}P(T>\vartheta_k|M=m)P(M=m)\nonumber\\
&=\sum_{m=1}^{n}s_m^{\tau} P(\sum_{i=1}^{k}W_i<m|M=m)\nonumber\\
&=\sum_{m=1}^{n}s_m^{\tau} P(\sum_{i=1}^{k}W_i<m)\nonumber\\
&=\sum_{m=1}^{n}s_m^{\tau} H_{k}(m-1)=\beta_{k,n},
\end{align}
where  the first equality follows from the fact that the lifetime of network is more than the arrival time of the $k$th shock if and only if the number of failed components in the time of $k$th shock is less than $m$ and  the second equality follows  because  the random variable $M$ is independent of $W_1,W_2,\dots.$ Since  $\vartheta_0\equiv 0$, and the network fails finally with one of the shocks, we have
\[\beta_{0,n}=1, \qquad \lim_{k\rightarrow \infty}\beta_{k,n}=0.\]
  On the other hand, since $\{T>\vartheta_{k+1}\}\subseteq \{T>\vartheta_k\}$, we get $\beta_{k+1}\leq \beta_k$ and hence $\beta_{k,n}$ is  decreasing in $k$.  Thus  $\beta_{k,n}$, as a function of $k$,  $k=0,1,\dots$, has properties of a discrete survival function.
 Let  ${\bf b}_n=(b_{1,n},b_{2,n},...)$ be  the probability mass function corresponding to $\beta_{k,n}$. That is,  $b_{k,n}=\beta_{k-1,n}-\beta_{k,n}$. Then, based on (\ref{eee}), we have
 \begin{align*}
b_{k,n}&= P(T>\vartheta_{k-1})-P(T>\vartheta_k)=P(T=\vartheta_k).
\end{align*}
\end{proof}
 From Proposition \ref{pro}, the $k$th element in ${\bf b}_n$, $b_{k,n}$, denotes  the probability that the network fails at the time of occurrence of the $k$th shock, $\vartheta_k$. We call, throughout the paper, the vector ${\bf b}_n$ as the vector of shock t-signature (ST-signature) of the network.

In the following, we show that the reliability function of the network lifetime can be represented as the reliability functions of epoch times $\vartheta_{i}$. For the counting process $\{\xi(t),t>0\}$, it is known that $\{\xi(t)=k\}$ if and only if $\{\vartheta_k\leq t <\vartheta_{k+1}\}$ where  $\vartheta_0\equiv 0$. Using this fact, we have
       \begin{align}
P(T>t)=& \sum_{k=0}^\infty \beta_{k,n} P(\xi(t)=k) \nonumber\\
=& \sum_{k=0}^\infty \beta_{k,n} P(\vartheta_k\leq t <\vartheta_{k+1})\nonumber\\
=& \sum_{k=0}^\infty \beta_{k,n} \left(P(\vartheta_{k+1}> t )-P(\vartheta_{k}> t)\right)\nonumber\\
=& \sum_{k=1}^\infty \beta_{k-1,n} P(\vartheta_k>t)-\sum_{k=1}^\infty \beta_{k,n} P(\vartheta_k>t)\nonumber\\
=& \sum_{k=1}^{\infty} b_{k,n} P(\vartheta_k>t).\label{rep-t1}
\end{align}

\begin{Remark}
{\rm  The model in \eqref{eq-org}, which arises in reliability theory, is known as the damage shock model (see \cite{barlow}, p. 92).   Let a device be subject to shocks appearing  randomly over time. Assuming that the device has a probability $\bar{P}(k)$ of surviving the first $k$ shocks, $k=0,1,...$, and $N(t)$ denotes the number of shocks that the device is subject to in the  interval $[0,t]$, then the reliability of the device, $\bar{H}(t)$, at time $t$  is
\begin{eqnarray*}
\bar{H}(t)=\sum_{k=0}^{\infty}P(N(t)=k)\bar{P}(k), \quad t\geq 0.
\end{eqnarray*}
Various properties of this model have been  explored  by different  authors; see, for example, \cite{ebrahimi}-\cite{pell1}. }
\end{Remark}

The hazard (failure) rate of a random variable $X$ or its distribution $F$ with density function $f$ is defined by $\lambda_F(x)={f(x)}/{\bar{F}(x)}$, where $\bar{F}=1-F$ is the survival function of $X$. The distribution function $F$  is said to be increasing hazard rate (IHR) if ${\bar{F}(t+x)}/{\bar{F}(t)}$ is decreasing in $t$ whenever $x>0$. From representation (\ref{rep-t1}), the hazard rate of the network can be written as
 \begin{eqnarray*}
\lambda(t)=\sum_{k=1}^{\infty} p_{k,n}(t) \lambda_k(t),
 \end{eqnarray*}
where $\lambda_k(t)$ is the hazard rate of $\vartheta_k$ and
\begin{eqnarray*}\label{pk}
p_{k,n}(t)=\frac{b_{k,n} P(\vartheta_k>t)}{\sum_{j=1}^{\infty} b_{j,n} P(\vartheta_j>t)}.
\end{eqnarray*}
It is interesting to note that $p_{k,n}(t)$ can be written as $p_k(t)=P(T=\vartheta_k|T>t).$ This is true because
\begin{align*}
P(T=\vartheta_k|T>t)&=\frac{P(T>t|T=\vartheta_k)P(T=\vartheta_k)}{P(T>t)}\\
&=\frac{b_{k,n} P(\vartheta_k>t)}{\sum_{j=1}^{\infty}b_{j,n} P(\vartheta_j>t)}= p_{k,n}(t),
 \end{align*}
where the second equality follows from the fact that  $\{\vartheta_k>t\}$ and $\{T=\vartheta_k\}$ are independent.

In the following, we make some stochastic comparisons between the performance of two networks,  where the components of the networks are subject to failure according to different or same counting processes.  We first use the following ordering definitions.
\begin{Definition} Let $X$ and $Y$ be two random variables with survival functions $\bar{F}$ and $\bar{G}$ having
density functions $f$ and $g$.
\begin{itemize}
\item[(a)] $X$ or $F$ is said to be stochastically less than or equal to $Y$ or $G$, denoted by $X\leq_{st} Y$ or $F\leq_{st}G$, if $\bar{F}(x)\leq \bar{G}(x)$ for all $x$.
\item[(b)] $X$ or $F$ is said to be less than or equal to a random variable $Y$ or $G$ in hazard rate order, denoted by $X\leq_{hr} Y$ or $F \leq_{hr} G$,  if $\dfrac{\bar{G}(x)}{\bar{F}(x)}$ increases in $x$.
\item[(c)] $X$ or $F$ is said to be less than or equal to a random variable $Y$ or $G$ in likelihood ratio order,
   denoted by $X \leq_{lr}Y$ or $F\leq_{lr} G$, if $\dfrac{g(x)}{f(x)}$ is an increasing function of $x$.
\end{itemize}
\end{Definition}
We have now the following theorem.
\begin{Theorem}\label{th-st-gen}
Consider two networks  consisting of $n_1$ and $n_2$ components and lifetimes $T_1$ and $T_2$, respectively.
Suppose that the  components of the $i$th network are subject  to shocks which  appear according to  counting  process $\{\xi_i(t),t\geq 0\}$, $i=1,2$.  Let  the ${\bf b}^{(i)}_{n_i}=( b^{(i)}_{1,n_i}, b^{(i)}_{2,n_i},...)$, $i=1,2$  denote the ST-signature of the $i$th network. If $\xi_1(t)\geq_{st}\xi_{2}(t)$ and ${\bf b}^{(1)}_{n_1}\leq_{st}{\bf b}^{(2)}_{n_2}$ then $T_1\leq_{st} T_2$.
\end{Theorem}
\begin{proof}
 Take $\beta_{k,n_i}^{(i)}=\sum_{j=k+1}^{\infty}b_{j,n_i}^{(i)}$, $i=1,2$.  Then, using \eqref{eq-org}, we have
 \begin{align*}
P(T_1>t)&=\sum_{k=0}^{\infty} \beta^{(1)}_{k,n_1} P(\xi_1(t)=k)\\
&\leq \sum_{k=0}^{\infty} \beta^{(1)}_{k,n_1} P(\xi_2(t)=k)\\
&\leq \sum_{k=0}^{\infty} \beta^{(2)}_{k,n_2} P(\xi_2(t)=k)=P(T_2>t),
\end{align*}
where  the first inequality follows from the facts that $\beta_{k,n_i}$,  $i=1,2$ is decreasing in $k$ and the assumption $\xi_1(t)\geq_{st}\xi_{2}(t)$. The second inequality  follows from the assumption that ${\bf b}^{(1)}_{n_1}\leq_{st}{\bf b}^{(2)}_{n_2}$.
\end{proof}

\begin{Corollary}\label{th-st-order}
In Theorem \ref{th-st-gen}, assume that the components of the two networks are subject to failure  by shocks appear according to renewal  processes $\{\xi_1(t),t\geq 0\}$ and $\{\xi_2(t),t\geq 0\}$, respectively. Let $X_{i,j}$, $i=1,2$, $j=1,2,...$, denote the time between the $(j-1)$th and $j$th shocks in the $i$th network.  Then the result of the theorem remains valid if we replace the condition  $\xi_1(t)\geq_{st}\xi_{2}(t)$ with  $X_{1,1}\leq_{st}X_{2,1}$.
 \end{Corollary}
 \begin{proof}
 Let $\vartheta_{i,k}=\sum_{j=1}^{k}X_{i,j},i=1,2,k=1,2,...$. Using Theorem 1.A.3 (b) of \cite{shaked}, the condition $X_{1,1}\leq_{st}X_{2,1}$ implies that $\vartheta_{1,k}\leq_{st}\vartheta_{2,k}$. Now the result follows from Theorem \ref{th-st-gen} and the fact that  for any counting process $\{\xi(t),t\geq 0\}$ with occurrence times $\vartheta_1,\vartheta_2,\dots$, we have  $\{\vartheta_n\leq t\}$ if and only if $\{\xi(t)\geq n\}$.
\end{proof}

Before presenting  the next theorem, we give the following definition (see, \cite{karlin}).
\begin{Definition}\label{def2}
    {\rm Let $A$ and $B$ be two subsets of the real line. A non-negative function $K$ defined on $A\times B$ is said to be totally positive of order 2,  denoted TP$_2$, if for all $a_1 < a_2$, and $b_1 < b_2$, ($a_i \in A$, $b_i \in B$, $i=1,2$),
    $$K(a_2,b_2)K(a_1,b_1)-K(a_1,b_2)K(a_2,b_1)\geq 0.$$ }
  \end{Definition}
In the next theorem, we show   when ST-signature vectors  of two networks are {\it hr} ordered then the lifetimes of the networks are also ordered in {\it hr} ordering.
\begin{Theorem}\label{th-hr-gen}
 Assume that the assumptions of Theorem \ref{th-st-gen} are  met and  that   the  components of  two  networks  are  subject  to failure by shocks appear  according to the same counting  process $\{\xi(t),t\geq 0\}$.  
If  ${\bf b}^{(1)}_{n_1}\leq_{hr}{\bf b}^{(2)}_{n_2}$ and  $P(\xi(t)=k)$ is $TP_2$ in $k\in\{0,1,\dots\}$ and $t>0$,  then $T_1\leq_{hr}T_2$.
\end{Theorem}
\begin{proof}  Let  $\beta_{k,n_i}^{(i)}=\sum_{j=k+1}^{\infty}b_{j,n_i}^{(i)},i=1,2$.
The assumption ${\bf b}^{(1)}_{n_1}\leq_{hr}{\bf b}^{(2)}_{n_2}$ implies that ${\beta_{k,n_2}^{(2)}}/{\beta_{k,n_1}^{(1)}}$ is increasing in $k$. Then, according to Definition \ref{def2}, it can be concluded that $\beta_{k,n_i}^{(i)}$ is $TP_2$ in $k\in\{0,1,...\}$ and $i\in{1,2}$.
Thus, if  $P(\xi(t)=k)$ is $TP_2$ in $k\in\{0,1,...\}$ and $t>0$, then from basic decomposition formula (see, \cite{karlin}), we get
$$P(T_i>t)=\sum_{k=1}^{\infty} \beta^{(i)}_{k,n_i} P(\xi(t)=k)$$
is $TP_2$ in $i\in\{1,2\}$, and $t>0$ which in turn implies that $T_1\leq_{hr}T_2$.
\end{proof}

\begin{Remark}
{\rm In Theorem \ref{th-st-gen}, if we assume that the components of two networks fail  by shocks appear according to the same renewal  processes based on i.i.d. r.v.s  $X_{i},i=1,2,...$, then under the assumption that  ${\bf b}^{(1)}_{n_1}\leq_{hr} {\bf b}^{(2)}_{n_2} $  and that $X_{1}$ has increasing hazard rate,  we have $T_1\leq_{hr}T_2$. This is true  because when $X_{1}$ has increasing hazard rate then $\vartheta_{k}\leq_{hr}\vartheta_{k+1},k=1,2,..$ and hence, the required result follows from the representation  (\ref{rep-t1}) and Theorem 1.B.14 of \cite{shaked}. Also if   $X_{1}$ has log-concave density function, then  $\vartheta_{k}\leq_{lr}\vartheta_{k+1},k=1,2,...$. Thus using  Theorem 1.C.17 of \cite{shaked}, if   ${\bf b}^{(1)}_{n_1}\leq_{lr} {\bf b}^{(2)}_{n_2} $ and $X_{1}$ has log-concave density function then   $T_1\leq_{lr}T_2$. }
\end{Remark}

\section{A binomial based model}
 In this section, we consider the shock model is presented in Section 2 and assume that   the number of component failures at each shock follows a binomial distribution.  Suppose that when a shock arrives each component fails with probability $p$.  Assuming  that the components fail independent of each other,  the number of failed components in the first shock, $W_1$, has binomial distribution $b(n,p)$, where $n$ is the number of components in the network. Suppose that, the number of failed components in the $i$th shock, $W_i$, $i\geq 2$, depends only on  $W_1,...,W_{i-1}$ through $\sum_{j=1}^{i-1}W_j$ and has binomial distribution $b(n_i,p)$, where $n_i=n-\sum_{j=1}^{i-1}W_j$. In other words,  assume that
\begin{eqnarray}\label{bin1}
P(W_1=k)={n \choose k} p^k q^{n-k},~~~k=0,1,...,n
\end{eqnarray}
and for $i\geq 2$,
\begin{eqnarray}\label{bin2}
P(W_i=k|\sum_{j=1}^{i-1}W_j=w)={n-w \choose k} p^k q^{n-w-k},~~~k=0,...,n-w,~w<n,
\end{eqnarray}
where  $q=1-p$.

Now  we can prove  the following lemma.
\begin{Lemma}\label{lem1}
Under the assumptions  (\ref{bin1}) and (\ref{bin2}), we have
$$P(\sum_{i=1}^{n}W_i=j)={n \choose j} (1-q^k)^j q^{k(n-j)},~~~j=0,...,n,~k=1,2,\dots.$$
\end{Lemma}

\begin{proof}
We prove the lemma by induction. For $k=1$ the result is true by relation (\ref{bin1}). Assume that the result is true for $k=m$. That is
$$ P(\sum_{i=1}^{m}W_i=j)= {n \choose j}(1-q^m)^{j} q^{m(n-j)}.$$
Then, for $k=m+1$, we get
\begin{eqnarray*}
P(\sum_{i=1}^{m+1}W_{i}=j)&=& \sum_{k=0}^{j} P(W_{m+1}=j-k|\sum_{i=1}^{m}W_i=k)P(\sum_{i=1}^{m}W_i=k)\\
&=&\sum_{k=0}^{j}  {n-k \choose j-k} {n \choose k} p^{j-k} q^{n-j} (1-q^m)^k q^{m(n-k)}\\
&=&\sum_{k=0}^{j}  {n-k \choose j-k} {n \choose k} p^{j-k} q^{n-j} p^k (\sum_{i=0}^{m-1} q^i)^k q^{m(n-k)}\\
&=&  {n \choose j}p^{j} q^{n(m+1)-j}\sum_{k=0}^{j}  {j \choose k}   (\frac{\sum_{i=0}^{m-1} q^{i}}{q^m})^k\\
&=& {n \choose j}p^{j} q^{n(m+1)-j} (\frac{\sum_{i=0}^{m} q^{i}}{q^m})^j\\
&=& {n \choose j}(1-q^{m+1})^{j} q^{(m+1)(n-j)},
 \end{eqnarray*}
 which is the required result.
\end{proof}
Now, based on the model given in (\ref{eq-org}), the reliability of the network  at time $t$ is
\begin{align}\label{bin}
P(T>t)=\sum_{k=0}^{\infty} \beta^*_{k,n} P(\xi(t)=k),
\end{align}
where $\beta^*_{0,n}=1$, and for $k=1,2,...$
\begin{align}
\beta^*_{k,n}&= \sum_{j=0}^{n-1} {\bar S}_j^{\tau}  {n \choose j} (1-q^k)^j q^{k(n-j)}\label{beta*}\\
&=\sum_{i=1}^n s_i^{\tau} \sum_{j=0}^{i-1} {n \choose j} (1-q^k)^j q^{k(n-j)}\nonumber \\
&=\sum_{m=1}^{n}\sum_{j=n-m}^{n-1}{\bar S}_j^{\tau}  {n \choose j} {j\choose n-m} (-1)^{j-n+m} q^{km}\nonumber.
\end{align}

From representation \eqref{rep-t1}, we have
\begin{eqnarray*}
P(T>t)=\sum_{k=1}^{\infty} b^*_{k,n} P(\vartheta_k>t),
\end{eqnarray*}
where $b^*_{k,n}=(\beta^*_{k-1,n}-\beta^*_{k,n})$.

In the following, we concentrate on a special case where the shocks appear as a nonhomogeneous Poisson process (NHPP). Recall that a counting process $\{\xi(t),t\geq 0\}$ is called a  NHPP if the survival function of arrival time $\vartheta_k$ of the $k$th event is
\begin{align*}
\bar{G}_{k}(t)=\sum_{x=0}^{k-1}\frac{[\Lambda(t)]^x}{x!}e^{-\Lambda(t)}, \ \quad t>0,\ k=1,2,...,
\end{align*}
where $\Lambda(t)=E(N(t))=-\log\bar{G}(t)$, and $\bar{G}(t)$ is the reliability function of the time to the first event. The function $\Lambda(t)$ is called the mean value function (m.v.f.). For more details on the properties of NHPP and related processes, one can see, for example, \cite{Nakagawa1}.

Let us look at the following example.
\begin{Example}
{\rm Consider a series network consisting of $n$ components. Suppose that the network is subject to shocks which appear according to a NHPP with m.v.f. $\Lambda(t)=-\log\bar{G}(t)$. Then under model (\ref{bin}) and noting that the t-signature of a series network is  ${\bf s}^{\tau}=(1,0,0,\dots,0)$, we can  easily see that
\[
\beta^*_{k,n}=q^{kn}, \qquad k=0,1,2,\dots.
\]
Hence,  the reliability of series network is given by
\begin{align*}
P(T>t)=&\sum_{k=0}^{\infty} q^{kn} P(\xi(t)=k)\\
=& \sum_{k=0}^{\infty} q^{kn} e^{-\Lambda(t)} \frac{\left(\Lambda(t)\right)^k}{k!}\\
=& e^{-\Lambda(t)} \sum_{k=0}^{\infty}  \frac{\left( q^{n} \Lambda(t)\right)^k}{k!}\\
=&e^{-\Lambda(t) \left(1-q^n\right)}\\
=&(\bar{G}(t))^{1-q^n}
\end{align*}
Note that if $n$, the number of components of the network, gets large then the reliability of the network tends to $\bar{G}(t)$.}
\end{Example}

In the sequel, we explore some aging  properties of the network lifetime.  First,  recall that a distribution $F$ is said to be increasing hazard rate average (IHRA) if $\left(\bar{F}(t)\right)^{1/t}$ is decreasing in $t>0$. It is well known that the IHR property implies the IHRA  (see \cite{barlow}).

We have the following lemma.

\begin{Lemma}\label{lem2}
$\beta^*_{k,n}$ is IHRA.
\end{Lemma}
\begin{proof} In order to prove the result, we must show $(\beta^*_{k,n})^\frac{1}{k}$ is decreasing in $k$ for $k=1,2,\dots$.
Note that  $\beta_{k,n}^{*}$ can be  rewritten as
\begin{align}\label{beta}
    \beta_{k,n}^{*}= \sum_{j=1}^{n} s_{j}^{\tau}\int_{1-q^k}^1 \frac{u^{j-1}(1-u)^{n-j}}{B(j,n-j+1)}du,
\end{align}
which is clearly an increasing function of $q^k$. It is clear from (\ref{beta*}) that  $\beta^*_{k,n}$ is a static reliability function of a network. If we write  $\beta^*_{k,n}=h(q^k)$,  where $h$ is  the reliability function of the network,  then
by choosing  $\alpha=\frac{k}{k+1}$ in Theorem 2.5 of Section 4 of  \cite{barlow}, we conclude that
$$h(q^{(k+1)(\frac{k}{k+1})})\geq h^{\frac{k}{k+1}}(q^{k+1})$$
which is equivalent to say that
$$(\beta^*_{k,n})^{\frac{1}{k}}\geq (\beta^*_{k+1,n})^{\frac{1}{k+1}}.$$
This completes the proof of the lemma.
\end{proof}

The following example shows that, although $\beta^*_{k,n}$  is always IHRA,  but it is not  necessarily IHR.
\begin{Example}\label{examp1} \em
Consider a bridge network pictured in Figure \ref{bridge}. It can be seen that the t-signature of this network is as ${\bf s}^{\tau}=(0,\frac{77}{270},\frac{154}{270},\frac{39}{270},0)$. In order to show that  $\beta^*_{k,n}$  is IHR we have to show, based on the definition of IHR distributions,  that  $\frac{\beta_{k+1}^*}{\beta_k^*}$ is decreasing in $k$.
\begin{figure}[h]
\begin{center}
\begin{tikzpicture}\label{rr}[node distance = 5 cm]
\tikzset{LabelStyle/.style =   {semithick,fill= white, text=black}}
\node[scale=0.9,semithick,shape = circle,draw, fill= black, text= white, inner sep =2pt, outer sep= 0pt, minimum size= 5 pt](A) at (0,0) {a};
\node[scale=0.8,semithick,shape = circle,draw, fill= white, text= black, inner sep =2pt, outer sep= 0pt, minimum size= 5 pt](B) at (1.5,1.5) {b};
\node[scale=0.9,semithick,shape = circle,draw, fill= white, text= black, inner sep =2pt, outer sep= 0pt, minimum size= 5 pt](C) at (1.5,-1.5) {c};
\node[scale=0.8,semithick,shape = circle,draw, fill= black, text= white, inner sep =2pt, outer sep= 0pt, minimum size= 5 pt](D) at (3,0) {d};
     \draw[semithick] (A) to node[LabelStyle]{1} (B) ;
     \draw[semithick] (A) to node[LabelStyle]{2} (C);
     \draw[semithick] (B) to node[LabelStyle]{3} (C);
     \draw[semithick] (B) to node[LabelStyle]{4} (D);
     \draw[semithick] (C) to node[LabelStyle]{5} (D);
\end{tikzpicture}
\end{center}
\caption{\small The bridge network.}
\label{bridge}
\end{figure}
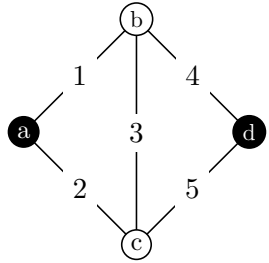

Figure \ref{fig:Fig1} shows the plot of $\frac{\beta_{k+1}^*}{\beta_k^*}$ for this network where  $q=0.5$. As the plot shows, this ratio is not decreasing for all values of $k$, hence $\beta_k^*$ is not IHR.%
\begin{figure}[h!]
  \centering
  \includegraphics[scale=0.4]{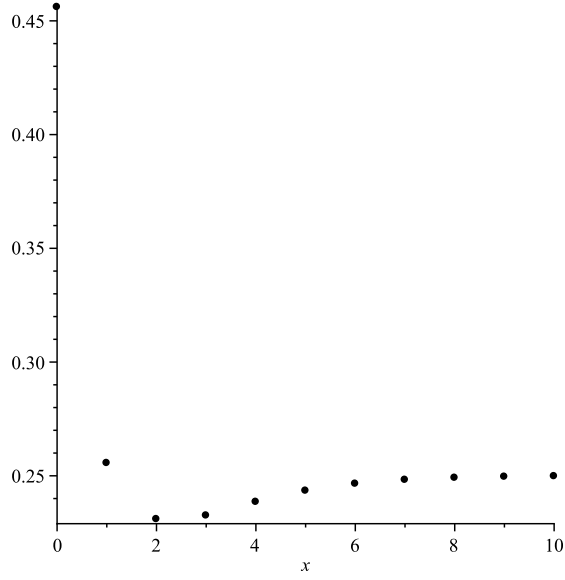}
  \caption{The plot of  $\frac{\beta^*_{k+1}}{\beta_k^*}$ for the bridge network.}
  \label{fig:Fig1}
\end{figure}
\end{Example}

Theorem 4.1 of \cite{got} implies that if  $P(\xi(t)=k)$ is $TP_2$ in $t\in(0,\infty)$, and $k\in\{0,1,...\}$ and $(E(a^{\xi(t)}))^{\frac{1}{t}}$ is  decreasing in $t$ for $a\in(0,1)$, then based on the fact that $\beta^{*}_{k,n}$ is IHRA we get that  $T$ is also IHRA. The following theorem shows that under the condition that  $\xi(t)$ is NHPP,  the network lifetime $T$ is IHRA if the distribution function of the arrival time of the first shock is IHRA.
\begin{Theorem}
Consider a network consisting of $n$ components with lifetime $T$. Suppose that the components of the network is subject to failure by shocks that appear according to a NHPP with m.v.f. $\Lambda(t)=-\log\bar{G}(t)$.  If $\bar{G}$ is IHRA, then $T$ is  IHRA.
\end{Theorem}
\begin{proof}
 For $k_1\leq k_2$,
\[\frac{P(\xi(t)=k_2)}{P(\xi(t)=k_1)}=\frac{k_1!}{k_2!}\left(\Lambda(t)\right)^{k_2-k_1}\]
is increasing in $t$ and hence $P(\xi(t)=k)$ is TP$_2$ in $k$ and $t$.  On the other hand, for $a\in (0,1)$
\begin{align}
E(a^{\xi(t)})=&\sum_{n=0}^{\infty} a^n \frac{(\Lambda(t))^n}{n!}e^{-\Lambda (t)}\nonumber\\
=& e^{-\Lambda(t)} \sum_{n=0}^{\infty}  \frac{(a\Lambda(t))^n}{n!}=e^{-{\Lambda(t)}(1-a)} \nonumber\\
=&\left(\bar{G}(t)\right)^{1-a}.\nonumber
\end{align}
If $\bar{G}(t)$ is IHRA, then $(G(t))^\frac{1}{t}$ is decreasing in $t$ and hence \[\left(E(a^{\xi(t)})\right)^{\frac{1}{t}}=\left(\bar{G}(t)\right)^\frac{1-a}{t}\]
 is decreasing in $t$.  Hence the result follows from Theorem 4.1 of Gottlieb \cite{got}.
\end{proof}

In the next theorem the stochastic relationships between t-signature vectors and the lifetimes of two networks are investigated.
\begin{Theorem}\label{th-or}
Consider two networks with lifetimes $T_1$ and $T_2$ and t-signature vectors  ${\bf s}_{1}^{\tau}=(s_{1,1}^{\tau},...,s_{1,n}^{\tau})$ and ${\bf s}_{2}^{\tau}=(s_{2,1}^{\tau},...,s_{2,n}^{\tau})$, respectively. Suppose that the  components of the $i$th network is subject  to failure by shocks appear according to  NHPP with m.v.f. $\Lambda_i(t)=-\log\bar{G}_i(t)$, $i=1,2$.   Assume that, upon arriving the shocks,  the components of the $i$th network fail  with probability $p_{i}$, $i=1,2$.
\begin{itemize}
\item[(a)] If $p_1\geq p_2$, $G_1 \leq_{st} G_2$ and ${\bf s}_{1}^{\tau}\leq_{st}{\bf s}_{2}^{\tau}$ then $T_1\leq_{st}T_2$.
\item[(b)] If $p_1=p_2$, $G_1=_{st}G_2$ and ${\bf s}_{1}^{\tau}\leq_{hr}{\bf s}_{2}^{\tau}$ then $T_1 \leq_{hr} T_2$.
\end{itemize}
\end{Theorem}
\begin{proof}
Let  $\beta_{k,n}^{*(i)}=\sum_{j=k+1}^{\infty}b_{j,n}^{*(i)},$ where $b_{j,n}^{*(i)}$ is the $j$th  element of  ST-signature associated to the $i$th network and $q_i=1-p_i,~i=1,2$.
 \begin{itemize}
  \item[(a)] Let $\{\xi_i(t),t>0\}$ be the  NHPP with m.v.f. $\Lambda_i(t)=-\log\bar{G}_i(t),~i=1,2$. Supose that $g_{j,n}(q_i^k)=\int_{1-q_i^k}^1 \frac{u^{j-1}(1-u)^{n-j}}{B(j,n-j+1)}du,~i=1,2.$
  Using (\ref{beta}), it can be seen that
 $\beta^{*(i)}_{k,n}=\sum_{j=1}^{n} s_{i,j}^{\tau} g_{j,n}(q_i^k)$, where $g_{j,n}(q_i^k)$ is an increasing function of $q_i^k$. Hence
\begin{align*}
 \beta_{k,n}^{*(1)}&= \sum_{j=1}^{n} s_{1,j}^{\tau} g_{j,n}(q_1^k)\\&\leq \sum_{j=1}^{n} s_{2,j}^{\tau} g_{j,n}(q_1^k)\\
  &\leq \sum_{j=1}^{n} s_{2,j}^{\tau} g_{j,n}(q_2^k)=\beta_{k,n}^{*(2)}
\end{align*}
in which the first inequality follows from the fact that ${\bf s}_1^{\tau}\leq_{st} {\bf s}_2^{\tau}$ and $g_{j,n}(q^k)$ is increasing in $j$ and second equality follows from the assumption $p_1\geq p_2$ which implies $g_{j,n}(q_1^k)\leq g_{j,n}(q_2^k)$. Also, $G_1 \leq_{st} G_2$ implies $\xi_1(t)\geq_{st} \xi_2(t)$.
Then the result follows from Theorem  \ref{th-st-gen}.
\item[(b)] It is easy to see that $ {n \choose j} (1-q_1^k)^j q_1^{k(n-j)}$ is $TP_2$ in $k$ and $j$. Also, ${\bf s}_1^{\tau}\leq_{hr} {\bf s}_2^{\tau}$ implies that ${\bar S}_{i,j}^{\tau} $ is $TP_2$ in $i$ and $j$. Therefore from basic decomposition formula (see \cite{karlin}),
     $$\beta^{*(i)}_{k,n}= \sum_{j=0}^{n-1} {\bar S}_{i,j}^{\tau}  {n \choose j} (1-q_1^k)^j q_1^{k(n-j)}$$ is $TP_2$ in $k\in\{0,1,...\}$ and $i\in\{1,2\}$ which implies ${\bf b}_{n}^{*(1)}\leq_{hr} {\bf b}_{n}^{*(2)}$. The proof is complete based on  Theorem \ref{th-hr-gen}.
\end{itemize}
\end{proof}

\begin{Example}\label{Example-dod}
{\rm Consider again Example \ref{examp1}. Let the network  be subject to shocks that appear according to a NHPP with m.v.f. $\Lambda(t)=-\log\bar{G}(t)$ and in each shock, each link fails with probability $0.1$.  We are interested in assessing  the reliability of the network in the cases where
the time to the first shock has either an exponential distribution with a constant hazard rate of $1$ (Exp(1)), or a Weibull distribution with  shape parameter $2$ and scale parameter $1$ (W(2,1)), or a linear hazard  distribution (L(1,1/2)). The survival functions of these distributions, respectively, are given as
\begin{align*}
\bar{G}_1(t)=&\exp(-t), \qquad t>0,\\
\bar{G}_2(t)=&\exp(-{t^2}), \qquad t>0,\\
\bar{G}_3(t)=&\exp(-t-t^2), \qquad t>0.
\end{align*}
\begin{figure}[h!]
  \centering
  \includegraphics[scale=0.42]{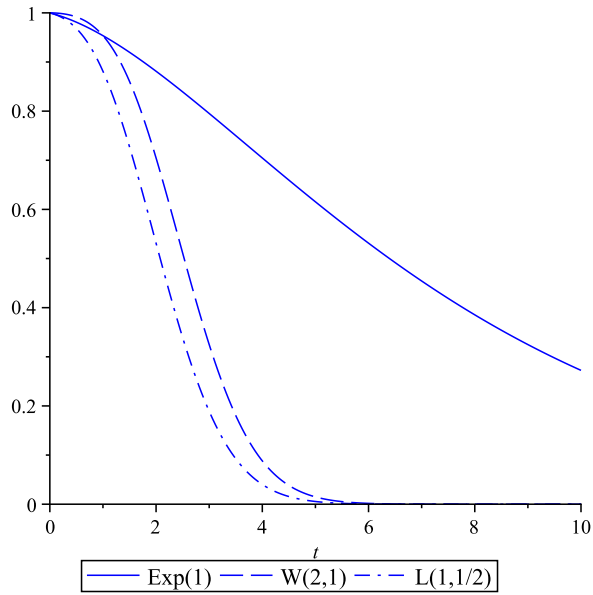}
  \caption{The plot of network reliability in Example \ref{Example-dod}.}
  \label{Fig}
\end{figure}
It can be easily shown that L(1,1/2) is stochastically less than both Exp(1) and W(2,1). Hence, as Figure \ref{Fig} reveals, based on Theorem \ref{th-or}, the reliability of the network for the L(1,1/2) case is less than that of  the cases of Exp(1) or W(2,1).
It can be easily seen that Exp(1) and W(2,1) are not stochastically ordered. Also, the plot shows that the network lifetimes are not stochastically ordered.}
\end{Example}
\section{Network reliability under fatal shocks}
In this section,  we assume that each shock is fatal for the network. That is, when a shock arrives it  leads to failure of  at least one component. Let fatal shocks occur according to a counting process, $\{\zeta(t),t>0\}$,   at random
time instants  $\varrho_1,\varrho_2,\dots$.  It is clear that the network finally fails by one of the fatal shocks.   In order to  obtain the reliability function of the network, in such a situation,  first we obtain $P(T=\varrho_i),~i=1,\dots n$.
Consider a network consists of $n$ components. It can be shown  that the  number of ways showing the order of component failures is  $n^*$ given in Lemma \ref{n*}.  Then,  under the
assumption that all ways of  the order of component failures are equally likely, we have
\begin{align*}
  s_i^*\equiv P(T=\varrho_i)=\frac{n_i  }{n^*},\quad i=1,\dots n,
\end{align*}
where $n_i$ is the  number of ways of the order of component failures in which $i$th fatal shock  causes the network fails.  It is obvious that $s_i^*$   just depends on the structure of the network. In the following example, we compute ${\bf s}^*=(s_1^*,\dots,s_n^*)$.
\begin{Example} {\rm
Consider again Example \ref{exa}. Let $\pi$ denote the order of link failures in the network and $r(\pi)$  the  shock number that caused the failure of the network.  All possible $\pi$ and corresponding $r(\pi)$ have been presented in Table \ref{tab2}. It is clear that
\begin{align*}
  s_1^*=P(T=\varrho_1)=\frac{7}{13} \qquad s_2^*=P(T=\varrho_2)=\frac{6}{13}, \qquad s_3^*=P(T=\varrho_3)=0.
\end{align*}
That is, ${\bf s}^*=(\frac{7}{13},\frac{6}{13},0)$.

\begin{table}[h]
\begin{center}
\caption{\small All possible $\pi$ and corresponding $r(\pi)$  }
\label{tab2}
\begin{tabular}{|c|c|c|c|c|c|}
\hline
 $\pi$ & $r(\pi)$ & $\pi$ & $r(\pi)$ & $\pi$ & $r(\pi)$ \\
 \hline
(1,2,3) & 1 &  (\{1,3\},2) & 1 & (\{1,2,3\}) & 1
\\ \hline
(1,3,2) & 1 &  (\{2,3\},1) & 1 &  &
\\ \hline
(2,1,3) & 2 &  (\{1,2\},3) & 1 &  &
\\ \hline
(2,3,1) & 2 &  (3,\{1,2\}) & 2 &  &
\\ \hline
(3,1,2) & 2 &  (2,\{1,3\}) & 2 &  &
\\ \hline
(3,2,1) & 2 &  (1,\{2,3\}) & 1 &  &
\\
\hline
\end{tabular}
\end{center}
\end{table}
}\end{Example}
From the fact that  $s_i^*$ does not depend
on the random mechanism of the component failures, we obtain the reliability function of the network as
\begin{align}\label{rep1}
P(T>t)=&\sum_{i=1}^{n} P(T>t|T=\varrho_i)P(T=\varrho_i)\nonumber\\
=& \sum_{i=1}^{n}s_i^*P(\varrho_i>t|T=\varrho_i)\nonumber\\
=& \sum_{i=1}^{n}s_i^*P(\varrho_i>t).
\end{align}
From the fact that $\zeta(t)=k$ if and only if $\varrho_k\leq t<\varrho_{k+1}$, it can be seen that
\begin{align}\label{rep2}
P(T>t)=&\sum_{i=0}^{n-1} {\bar S}^*_iP(\zeta(t)=i)
\end{align}
where ${\bar S}^*_i=\sum_{j=i+1}^{n}s^*_j$.

\begin{Remark}{\rm
It is noted that the representations (\ref{rep1}) and (\ref{rep2}) are similar to representations (\ref{rep-t1}) and (\ref{eq-org}), respectively. Hence, the results  obtained based on (\ref{rep-t1}) and (\ref{eq-org}) in Section \ref{2} are valid for the fatal shock model.
}\end{Remark}

{\bf Acknowledgement:}

 M. Asadi's research was carried out in IPM Isfahan branch and was in part supported by a grant from IPM (No. 93620411).
{\small
{}}
\end{document}